\documentclass[10pt,conference]{IEEEtran}
\usepackage{color}
\usepackage{amsfonts}
\usepackage{latexsym}
\usepackage{amsthm}
\usepackage{amssymb}
\usepackage{enumerate}
\usepackage{tikz,pgf,pgfplots}
\usepackage{cite}
\usepackage{bm}
\usepackage{bbm}
\usepackage{xcolor}
\usetikzlibrary{patterns}
\theoremstyle{definition}
\newtheorem{proposition}{\quad Proposition}  

\newtheorem{lemma}[proposition]{\quad Lemma}
\newtheorem{corollary}[proposition]{\quad Corollary}

\newtheorem{theorem}[proposition]{\quad Theorem}
\newtheorem{assumption}{\quad Assumption}

\theoremstyle{remark}

\ifCLASSINFOpdf

\else

\usepackage[dvips]{graphicx}

\fi

\usepackage[cmex10]{amsmath}

\begin{document}

\title{Optimal Energy-Efficient Regular Delivery of Packets in Cyber-Physical Systems}

\author{\IEEEauthorblockN{Xueying Guo\IEEEauthorrefmark{1}, Rahul Singh\IEEEauthorrefmark{2}, P. R. Kumar\IEEEauthorrefmark{2} and Zhisheng Niu\IEEEauthorrefmark{1}}
	\IEEEauthorblockA{\IEEEauthorrefmark{1}Department of Electronic Engineering, Tsinghua Uviersity, R. R. China\\
		Email: guo-xy11@mails.tsinghua.edu.cn, niuzhs@tsinghua.edu.cn}
	\IEEEauthorblockA{\IEEEauthorrefmark{2}Department of Electrical and Computer Engineering, Texas A\&M University, USA\\
		Email:\{rsing1, prk\}@tamu.edu	}
}

\maketitle

\begin{abstract}
In cyber-physical systems such as in-vehicle wireless sensor networks, a large number of sensor nodes continually generate measurements that should be
received by other nodes such as actuators in a regular fashion. Meanwhile, energy-efficiency is also important in wireless sensor networks. Motivated by these, we develop scheduling policies which are energy efficient and simultaneously maintain ``regular" deliveries of packets. A tradeoff parameter is introduced to balance these two conflicting objectives. 
We employ a Markov Decision Process (MDP) model where the state of each client is the time-since-last-delivery of its packet, and reduce it into an equivalent finite-state MDP problem. Although this equivalent problem can be solved by standard dynamic programming techniques, it suffers from a high-computational complexity.
Thus we further pose the problem as a restless multi-armed bandit problem and employ the low-complexity Whittle Index policy. 
It is shown that this problem is indexable and the Whittle indexes are derived. Also, we prove the Whittle Index policy is asymptotically optimal and validate its optimality via extensive simulations. 
\end{abstract}

\IEEEpeerreviewmaketitle

\section{Introduction}
Cyber-physical systems typically employ wireless sensors 
for keeping track of physical processes such as temperature and pressure.
These nodes then transmit data packets containing these measurements back to the access point/base station. Moreover, these packets 
should be delivered in a ``regular" way. 
So, time between successive deliveries of packets, i.e. inter-delivery time, is an important performance metric \cite{atilla1,atilla2}. 
Furthermore, many wireless sensors are battery powered. Thus, energy-efficiency is also important. 

We address the problem of satisfying these dual conflicting objectives: inter-delivery time requirement and energy-efficiency. We design wireless scheduling policies that support the inter-delivery requirements of such wireless clients in an energy-efficient way. 
In \cite{Singh2013,Singh2014}, the authors analyzed the growth-rate of service irregularities that occur for the case of multiple clients sharing a wireless network and when the system is in heavy traffic regime. The inter-delivery performance of the Max Weight discipline under the heavy traffic regime was studied in \cite{Singh2015Sig}.
To the authors' best knowledge, the inter-delivery time was first considered in 
\cite{atilla1,atilla2} as a performance metric for queueing systems, where a sub-optimal policy is proposed to trade off the stablization of the queues and service regularity. 
However, this is different from our problem, where the arrival process does not need to be featured. 
In our previous work \cite{Singh2015Infocom}, throughput is traded off for better performance with respect to variations in inter-delivery times. However, tunable and heterogeneous inter-delivery requirements have not been considered.


In this paper, we formulate the problem as a Markov Decision Process (MDP) with a system cost consisting of the summation of the penalty for exceeding the inter-delivery threshold and a weighted transmission energy consumption. An energy-efficiency weight parameter $\eta$ is introduced to balance these two aspects. To solve this infinite-state MDP problem, we reduce it to an equivalent MDP comprising of only a finite number of states. This equivalent finite-state finite-action MDP can be solved using standard dynamic programming (DP) techniques. 

The significant challenge of this MDP approach is the computational complexity, since the state-space of the equivalent MDP increases exponentially in the number of clients. To address this, we further formulate this equivalent MDP as a restless multi-armed bandit problem (RMBP), with the goal of exploiting a low-complexity index policy. 


In this RMBP, we first derive an upper bound on the achievable system reward by exploring the structure of a relaxed-constraint problem. Then, we determine the Whittle index for our multi-armed restless bandit problem, and prove that the problem is indexable. In addition, we show the resulting index policy is optimal in certain cases, and validate the optimality by a detailed simulation study. 
The impact of the energy-efficiency parameter $\eta$ is also studied in the simulation results.

\section{System Model}
Consider a cyber-physical system in which there are $N$ wireless sensors and one access point (AP). 
We will assume that time is discrete. 
At most $L$ sensors can simultaneously transmit in a time slot. 
In each time-slot, 
a control message is broadcasted at the beginning by the AP to inform which set of $L$ sensors can transmit in the current time-slot. Each of the assigned sensors then makes a sensor measurement and transmits its packet. 
The length of a time slot is the time required for the AP to send the control message plus the time required for the $L$ assigned clients to prepare and transmit a package. 

The wireless channel connecting the sensor and the AP is unreliable. When client $n$ is selected to transmit, it succeeds in delivering a packet with a probability $p_n\in (0,1)$. Furthermore, each attempt to transmit a packet of client $n$ consumes $E_n$ units of energy.

The QoS requirement of client $n$ is specified through an integer, the \textit{packet inter-delivery time threshold} $\tau_n$.  
The cost incurred by the system during the time interval $\{0,1,\ldots,T\}$ is given by, 
\begin{align}\label{eq:cost define in System Model}
\nonumber 
\mathrm{E} \Big[   \sum_{n=1}^N   \Big(\sum_{i=1}^{M_T^{(n)}} ( D_i^{(n\!)} - \tau_n)^+   
&+ (T - t_{ D_{ M_T^{(n)}}^{ (n )}} -\tau_n)^+ \\ 
&+\eta   \hat{M}_T^{(n)} E_n \Big)
  \Big] ,
\end{align}
where $D_i^{(n)}$ is the time between the deliveries of the $i$-th and $(i+1)$-th packets for client $n$, $M_T^{(n)}$ is the number of packets delivered for the $n$-th client by the time $T$, $t_{D_{i}^{(n)}}$ is the time slot in which the $i$-th package for client $n$ is delivered, 
$\hat{M}_T^{(n)}$ is the total number of slots in $\{0,1,\cdots,T\!-\!1\}$ in which the $n$-th client is selected to transmit, 
and $(a)^+:=\max\{a,0\}$. 
The second term is included since, otherwise, no transmission at all will result in the least cost. 
The last term weights the total energy consumption in $T$ time-steps by a non-negative \textit{energy-efficiency parameter} $\eta$, which tunes the weightage given to energy conservation.
The access point's goal is to select at most $L$ clients to transmit in each time-slot from among the $N$ clients, so as to minimize the above cost.
 
\section{Reduction to Finite State Problem}
In the following, vectors will be denoted bold font, i.e., $\mathbf{a}: =\left(a_1,\ldots,a_N\right)$. Define
$\mathbf{a}\wedge \mathbf{b}: = \left(a_1 \wedge b_1,\ldots, a_N\wedge b_N\right)$. 
Random processes will be denoted by capitals.

We formulate our system as a Markov Decision Process, as follows. 
The system state at time-slot $t$ is denoted by a vector $X(t):=\left(X_1\left(t\right),\cdots,X_N\left(t\right) \right)$, where $X_n(t)$ is the time elapsed since the latest delivery of client $n$'s packet.
Denote the action at time $t$ as $U(t):=\left( U_1\left(t\right),\cdots, U_N\left(t\right) \right)$,  with $\sum_{n=1}^N U_n(t) \leq L$ for each $t$, where
 \begin{align*}
 U_n(t)=
 \begin{cases}
 1 \mbox{ if client } n \mbox{ is selected to transmit in slot } t,\\
 0 \mbox{ otherwise}.
 \end{cases}
 \end{align*}
 The system state evolves as, 
\begin{align*}
X_n(t+1) =
\begin{cases}
0 \mbox{~~if a packet of client } n \mbox{ is delivered in } t,\\
X_n(t) + 1  \mbox{ otherwise. }
\end{cases}
\end{align*} 
Thus, the system forms a controlled Markov chain (denoted \textit{MDP-1}), with the transition probabilities given by, 
\begin{align*}
P_{\mathbf{x},\mathbf{y}}^\text{MDP-1}(\mathbf{u})
:= \mathrm{P}\left[X(t+1)=\mathbf{y}\big|X(t)=\mathbf{x},U(t)=\mathbf{u}\right] \\
= \prod_{n=1}^{N} 
	\mathrm{P}\left[X_n(t+1)=y_n \big|X_n(t)=x_n,U_n(t)=u_n\right], 
\end{align*}
with ~~~~ $\mathrm{P}\left[X_n(t+1)=y_n \big|X_n(t)=x_n,U_n(t)=u_n\right]$
\begin{align*}
:=
\begin{cases}
p_n & \mbox{if }y_n=0 \mbox{ and }u_n=1, \\
1-p_n & \mbox{if }y_n=x_n+1 \mbox{ and }u_n=1, \\
1 & \mbox{if }y_n=x_n+1 \mbox{ and }u_n=0, \\
0 & \mbox{otherwise.}
\end{cases}
\end{align*}

The $T$-horizon optimal cost-to-go from initial state $\mathbf{x}$ is given by, 
\begin{align*} 
V_T(\mathbf{x})&:=\min_{\pi: \sum_n\!\! U_n(t)\leq L} \mathrm{E}\bigg\{
\sum_{t=0}^{T-1} \sum_{n=1}^{N}  \Big(
\eta E_n U_n(t) \\ 
&+
\left(X_n(t)\! +\! 1\!-\!\tau_n \right)^{+} \bm{1}\!\left\{X_n\!\left(t\!+\!1\right) = 0\right\}
\Big) 
\bigg|X(0)=\mathbf{x} \bigg\}, 
\end{align*}
where 
$\bm{1}\{\cdot\}$ is the indicator function, and $X(T):=\mathbf{0}$ (which leads to recovering the second term in the cost \eqref{eq:cost define in System Model}), and the minimization is over the class of history dependent policies.

The Dynamic Programming (DP) (see \cite{Puterman1994}) recursion is,
\begin{align}\label{eq:recursive relationship original MDP-1}
\nonumber 
V_T(\mathbf{x})&=\min_{\mathbf{u}:\sum_n \!\! u_n \leq L} \mathrm{E}\bigg\{
\eta \sum_{n=1}^N E_n u_n + \sum_{\mathbf{y}} P_{\mathbf{x},\mathbf{y}}^\text{MDP-1}(\mathbf{u}) \\ &\cdot \left[\sum_{n=1}^N \left(x_n+1-\tau_n \right)^+ \bm{1}\!\left\{y_n=0 \right\} +V_{T-1}(\mathbf{y})\right]
 \bigg\}.
\end{align} 

The above problem, denoted as MDP-1, involves a countably infinite state space. The following results show that it can be replaced by an equivalent finite state MDP.

\begin{lemma}\label{th:fundamental claim}
For the MDP-1, we have, $~\forall x_1,\cdots,x_N \geq 0$, 
\begin{align*} 
V_T(x_1,\cdots,\tau_i\!+\!x_i,\cdots,x_N\!)=x_i\!+\!V_T(x_1,\cdots,\tau_i,\cdots,x_N\!).
\end{align*}
Moreover, the optimal actions for the states $(x_1,\cdots,\tau_i+x_i,\cdots,x_N)$ and $(x_1,\cdots,\tau_i,\cdots,x_N)$ are the same. 
\end{lemma}
\begin{proof}
Let us consider the MDP-1 starting from two different initial states, $\mathbf{x}=(x_1,\cdots,\tau_i+x_i,\cdots,x_N)$ and  $\tilde{\mathbf{x}}=(x_1,\cdots,\tau_i,\cdots,x_N)$, and compare their evolutions.
Construct the processes associated with both the systems on a common probability space and couple stochastically the successful transmissions for the two systems. 
Let $\pi$ be an arbitrary history-dependent policy that is applied to in the first system (starting in state $\mathbf{x}$).
Corresponding to $\pi$, there is a policy $\tilde{\pi}$ in the second system, which takes the same actions as the policy $\pi$ at each time slot. Then all the packet-inter-delivery times for both the processes are the same, except for the first inter-delivery time of the $i$-th client, which is larger for the former system as compared to the latter by $x_i$. 
In addition, 
Since the policy $\pi$ is arbitrary, 
$V_T(\mathbf{x})\geq  x_i+ V_T(\tilde{\mathbf{x}})$. 
The inequality in the other direction is proved similarly. 
The proof of the second statement follows by letting $\pi$ be the optimal policy. 
\end{proof}

\begin{corollary}\label{th:equivalent recur-rela corollary}
For any system state $\mathbf{x}$ such that $x_n\leq \tau_n,\forall n$,
\begin{align}\label{eq:recursive relationship new MDP-1}
\nonumber
V_T(\mathbf{x})=\min_{\mathbf{u}:\sum_n u_n\leq L}& \mathrm{E}\bigg\{\sum_n  \left(\eta E_n u_n +\bm{1}\left\{x_n=\tau_n\right\} \right)\\
&+ \sum_{\mathbf{y}}P_{\mathbf{x},\mathbf{y}}^\text{MDP-1} V_{T-1}(\mathbf{y}\wedge\boldsymbol{\tau})
 \bigg\}.
\end{align}
\end{corollary}
\begin{proof}
Consider the equation \eqref{eq:recursive relationship original MDP-1} and the following two cases:
\begin{enumerate}
\item The initial state $\mathbf{x}$ is such that $x_n < \tau_n, \forall n$. Then $(x_n+1-\tau_n)^+=0$ and $\bm{1}\{x_n=\tau_n \}=0$. 
In addition, for any action $\mathbf{u}$, if $\mathbf{y}$ is any state such that $P^{\text{MDP-1}}_{\mathbf{x},\mathbf{y}}(\mathbf{u})>0 $, then $\mathbf{y}$ satisfies $y_n\leq \tau_n,\forall n$, which shows,  $\mathbf{y}=\mathbf{y}\wedge\boldsymbol{\tau}$.

\item There exists an $i$ such that the initial state $\mathbf{x}$ satisfies $x_i=\tau_i$. Let us first assume there is only one client $i$ satisfying $x_i=\tau_i$ and that $x_j< \tau_j,\forall j\neq i$. 
Then, for any action $\mathbf{u}$, if $\mathbf{y}$ is any state such that $P^{\text{MDP-1}}_{\mathbf{x},\mathbf{y}}(\mathbf{u})>0 $, we have $y_j\leq \tau_j,\forall j\neq i$, and also $y_i$ is either $0$ or $\tau_i+1$. If $y_i=0$ and $y_j\leq \tau_j, \forall j\neq i$, then $(x_i+1-\tau_i)^+\bm{1}(y_i=0)=1$ and $\mathbf{y} =\mathbf{y}\wedge \boldsymbol{\tau}$. If $y_i=\tau_i+1$, and $y_j\leq \tau_j, \forall j\neq i$, then from Lemma \ref{th:fundamental claim}, $V_{T-1}(\mathbf{y})=1+V_{T-1}(\mathbf{y}\wedge \boldsymbol{\tau})$. Thus, when there is only one client $i$ satisfying $x_i=\tau_i$, the r.h.s (right-hand side) of  \eqref{eq:recursive relationship original MDP-1} can be rewritten as, 
\begin{align*}
\min_{\mathbf{u}:\sum_n\!\! u_n\leq L} \!\!\!\! \mathrm{E}\left\{\eta\!\sum_n\!\! E_n u_n   \!+\!1\!+\! \sum_{\mathbf{y}}P_{\mathbf{x},\mathbf{y}}^\text{MDP-1} V_{T-1}(\mathbf{y}\wedge\boldsymbol{\tau})
\right\}. 
\end{align*}
The case where there are one or more clients $j\neq i$ satisfying $x_j=\tau_j$ is proved similarly. 
\end{enumerate}
\end{proof}

The following lemma can be easily derived, the proof of which is omitted due to space constraints.
\begin{lemma} \label{th:Y(t) is MDP Lemma}
$Y(t):=X(t)\wedge \boldsymbol{\tau}$ is a Markov Decision Process 
\begin{align*}
\text{with~~~}&\mathrm{P}\left[Y(t+1)|Y(t),\cdots,Y(0),U(t),\cdots,U(0) \right]\\
=&\mathrm{P}\left[Y(t+1)|Y(t),U(t)\right]. 
\end{align*}
\end{lemma}


Now we construct another MDP, denoted \textit{MDP-2}, which is equivalent to the MDP-1 in an appropriate sense. We will slightly abuse notation and continue to use the symbols $Y(t)$ and $U(t)$ for states and controls.

For  
$Y_n(0) \in \{0,1,\cdots,\tau_n\}$,  let $Y_n(t)$ evolves as, 
\begin{align*}
Y_n(t+1)=
\begin{cases}
0 ~~\mbox{if a packet is delivered for client $n$ at $t$, } \\
\left(Y_n(t)+1\right) \wedge \tau_n ~~\mbox{otherwise.}
\end{cases}
\end{align*}
Denote by $P_{\!\mathbf{x},\mathbf{y}}^\text{MDP-2}$ the transition probabilities of the resulting process $Y(t):=\left(Y_1(t),\cdots,Y_N(t)\right)$ on the state space 
$\mathbb{Y}:=\prod_{n=1}^N\{0,1,\cdots,\tau_n\}$, where the transition probabilities, 
\begin{align}\label{eq:indivitual transition pro. MDP-2}
\nonumber 
\mathrm{P}\left[Y_n(t+1)=y_n \big|Y_n(t)=x_n,U_n(t)=u_n\right]\\
=
\begin{cases}
p_n & \mbox{if }y_n=0 \mbox{ and }u_n=1, \\
1-p_n & \mbox{if }y_n=(x_n+1)\wedge \tau_n  \mbox{ and }u_n=1, \\
1 & \mbox{if }y_n=(x_n+1)\wedge \tau_n  \mbox{ and }u_n=0, \\
0 & \mbox{otherwise.}
\end{cases}
\end{align} 

The optimal cost-to-go function for MDP-2 is, 
\begin{align}\label{eq:optimal value function MDP-2} 
\nonumber
V_T(\mathbf{x}):=&\!\!\!\!\min_{\pi: \sum_n\!\! U_n(t)\leq L} \!
 \mathrm{E}\bigg\{
\sum_{t=0}^{T\!-\!1} \sum_{n=1}^{N}  \bm{1}\!\left\{Y_n\!\left(t\right) = \tau_n \right\}
\\
&+ \eta E_n U_n(t) \bigg|Y(0)=\mathbf{x} \bigg\},
\forall \mathbf{x}\in \mathbb{Y}. 
\end{align}

\begin{theorem}
MDP-2 is equivalent to the MDP-1 in that:
\begin{enumerate}
\item MDP-2 has the same transition probabilities as the accompanying process of MDP-1, i.e., the process $X(t)\wedge \boldsymbol{\tau}$;
\item Both MDPs satisfy the recursive relationship in \eqref{eq:recursive relationship new MDP-1}; thus, their optimal cost-to-go functions are equal for each starting state $\mathbf{x}$ with $x_n\leq \tau_n,\forall n$; 
\item Any optimal control for MDP-1 in state $\mathbf{x}$ is also optimal for MDP-2 in state $\mathbf{x}\wedge \boldsymbol{\tau}$.
\end{enumerate}
\end{theorem}
\begin{proof}
Statement 1) directly follows Lemma \ref{th:Y(t) is MDP Lemma}. 
The DP recursion for the optimal cost in MDP-2 is
\begin{align}\label{eq:recursive relationship MDP-2}
\nonumber
V_T(\mathbf{x})=\min_{\mathbf{u}:\sum_n u_n\leq L}& \mathrm{E}\bigg\{\sum_n  \left(\eta E_n u_n +\bm{1}\left\{x_n=\tau_n\right\} \right)\\
&+ \sum_{\mathbf{y}}P_{\mathbf{x},\mathbf{y}}^\text{MDP-2} V_{T-1}(\mathbf{y})
\bigg\}.
\end{align}	
Thus, statement 2) is obtained from \eqref{eq:recursive relationship MDP-2} and Corollary \ref{th:equivalent recur-rela corollary}. 
In addition, statement 3) follows Lemma \ref{th:fundamental claim} and statement 1).  
\end{proof}
As a result, we focus on MDP-2 in the sequel. 

\section{Optimal Index Policy for the Relaxed Problem}
\label{se:asymptotically optimal section}
\subsection{Formulation of Restless Multi-armed Bandit Problem}
MDP-2, with a finite state space, can be solved in a finite number of steps by standard DP techniques (see \cite{Puterman1994}). However, even for a finite time-horizon, it suffers from high computational complexity, since the cardinality of the state space increases exponentially in the number $N$ of clients. 

To overcome this, we formulate MDP-2 as an infinite-horizon restless multi-armed bandit problem (\cite{Whittle1988,Weber1990}), and obtain an Index policy which has low complexity. 

We begin with some notations: Denote by $\alpha$ the \textit{maximum fraction} of clients that can simultaneously transmit in a time slot, i.e., $\alpha={L}/{N}$. 
The process $Y_n(t)$ associated with client $n$ is denoted as \textit{project} $n$ in conformity with the bandit nomenclature. If $U_n(t)=1$, the project $n$ is said to be \textit{active} in slot $t$; while if $U_n(t)=0$, it is said to be \textit{passive} in slot $t$. 

The infinite-horizon problem is to solve, with  $Y(0)=\mathbf{x}\in \mathbb{Y}$, 
\begin{align} \label{eq:original-OP object}
\max_{\pi}~ & \liminf_{T\rightarrow +\infty} \frac{1}{T}\mathrm{E}\Big[\!\!\sum_{t=0}^{T-1}
\!\!\sum_{n=1}^N\!\!\! -\bm{1}\{Y_n(t)=\tau_n\} -\!\eta E_n U_n\!(t)
 \Big] \\ \label{eq:original-OP constraint}
\mbox{s.t.} ~~ & \quad \quad \sum_{n=1}^N (1-U_n(t))\geq (1-\alpha)N, ~\forall t.
\end{align}
Note that
the system reward is considered instead of the system cost. 
\subsection{Relaxations}\label{se:relaxations}
We consider an associated relaxation of the problem \eqref{eq:original-OP object}-\eqref{eq:original-OP constraint} which puts a constraint only on the \textit{time average} number of active projects allowed:  
\begin{align} \label{eq:relaxed-OP object}
\max_{\pi}~ & \liminf_{T\rightarrow +\infty} \frac{1}{T}\mathrm{E}\!\!\left[\sum_{t=0}^{T-1}\!\!
\sum_{n=1}^N\!\!\! -\bm{1}\{Y_n\!(t)\!=\!\tau_n\!\} \!-\!\eta E_n U_n(t) 
\right] \\ \label{eq:relaxed-OP constraint}
\mbox{s.t.} ~~ & \liminf_{T\rightarrow+\infty} \frac{1}{T} \mathrm{E}\left[\sum_{t=0}^{T-1}\sum_{n=1}^N \left(1-U_n(t)\right) \right] \geq (1-\alpha)N.
\end{align}
Since constraint \eqref{eq:relaxed-OP constraint} relaxes the stringent requirement in \eqref{eq:original-OP constraint}, it provides an upper bound on the achievable reward in the original problem. 

Let us consider the Lagrangian associated with the problem \eqref{eq:relaxed-OP object}-\eqref{eq:relaxed-OP constraint}, with $Y(0)=\mathbf{x}\in \mathbb{Y}$, 
\begin{align*}
l(\pi,\omega)&:= \liminf_{T\rightarrow +\infty} \frac{1}{T}\mathrm{E}_{\pi}\!\!\left[\sum_{t=0}^{T-1} \!\! 
\sum_{n=1}^N \!\!\! -\bm{1}\{Y_n\!(t)\!=\!\tau_n\} \!-\!\eta E_n U_n(t) 
\right] \\
& +\omega \liminf_{T\rightarrow +\infty} \frac{1}{T} \mathrm{E}_{\pi}\!\!\left[\sum_{t=0}^{T-1}\sum_{n=1}^N \left(1\!-\!U_n(t)\right)\right] 
-\omega (1\!-\!\alpha)N,
\end{align*}
where $\pi$ is any history-dependent scheduling policy, while $\omega\geq 0$ is the Lagrangian multiplier. 
The Lagrangian dual function is $d(\omega):=\max_{\pi} l(\pi,\omega)$: 
\begin{align} \label{eq:dual function upper bound} 
\nonumber
d(\omega)&\leq \max_{\pi} ~
\liminf_{T\rightarrow+\infty} \frac{1}{T} 
\mathrm{E}\bigg[\sum_{t=0}^{T-1}
\sum_{n=1}^N -\bm{1}\{Y_n(t)=\tau_n\} \\ \nonumber
& -\eta E_n U_n(t) +\omega \left(1-U_n(t)\right) \bigg| Y(0)=\mathbf{x}
\bigg]
\!-\!\omega (1\!-\!\alpha)N \\ \nonumber
&\leq \max_{\pi} ~
\limsup_{T\rightarrow+\infty} \frac{1}{T}
\mathrm{E}\bigg[\sum_{t=0}^{T-1}
\sum_{n=1}^N -\bm{1}\{Y_n(t)=\tau_n\} \\ \nonumber
& -\eta E_n U_n(t) +\omega \left(1-U_n(t)\right) \bigg| Y(0)=\mathbf{x}
\bigg]
\!-\!\omega (1\!-\!\alpha)N \\ \nonumber
&\leq \max_{\pi} ~ \sum_{n=0}^N 
\limsup_{T\rightarrow+\infty} \frac{1}{T}
\mathrm{E}\bigg[\sum_{t=0}^{T-1} -\bm{1}\{Y_n(t)=\tau_n\} \\ 
&
-\eta E_n U_n\!(t) \!+\!\omega \left(1\!-\!U_n\!(t)\right) \bigg| Y\!(0)\!=\!\mathbf{x}
\bigg]
\!-\!\omega (1\!-\!\alpha)N,
\end{align}
where the first and the third inequalities hold because of the super/sub-additivities of the limit inf/sub (respectively). 

Now, consider the unconstrained problem in the last two lines of \eqref{eq:dual function upper bound}. It can be viewed as a composition of $N$ independent $\omega$-subsidy problems interpreted as follows:  
For each client $n$, besides the original reward $-\bm{1}\{Y_n(t)=\tau_n\} -\eta E_n U_n(t)$, when $U_n(t)=0$, it receives a subsidy $\omega$ for being passive. 

Thus, the $\omega$\textit{-subsidy problem} associated with client $n$ is defined as,
\begin{align}\label{eq:w-subsidy problem}
\nonumber 
R_n(\omega)&=\max_{\pi_n}~
\limsup_{T\rightarrow+\infty} \frac{1}{T}
\mathrm{E}\bigg[\sum_{t=0}^{T-1} -\bm{1}\{Y_n(t)=\tau_n\} \\ &-\eta E_n U_n(t) +\omega \left(1-U_n(t)\right) \bigg| Y_n(0)=x_n
\bigg],
\end{align}
where $\pi_n$ is a history dependent policy which decides the action $U_n(t)$ for client $n$ in each time-slot. 

In the following, we first solve this $\omega$-subsidy problem, and then explore its properties to show that strong duality holds for the relaxed problem \eqref{eq:relaxed-OP object}-\eqref{eq:relaxed-OP constraint}, and thereby determine the optimal value for the relaxed problem.

For $\theta\in \{0,1,\cdots,\tau_n\}$ and $\rho\in [0,1]$, we define $\sigma_n(\theta,\rho)$ to be a \textit{threshold policy} for project $n$, as follows: The policy $\sigma_n(\theta,\rho)$ keeps the project passive at time $t$ if  $Y_n(t)<\theta$. However when $Y_n(t)>\theta$, the project is activated, i.e., $U_n(t)=1$. If $Y_n(t)=\theta$, then at time $t$, the project stays passive with probability $\rho$, and is activated with probability $1-\rho$. 

For each project $n$, associate a function 
\begin{align}\label{eq:whittle index original in subsidy}
W_n(\theta):=p_n (\theta+1) (1-p_n)^{\tau_n-(\theta+1)}-\eta E_n,   
\end{align}
where $\theta ={0,1,\cdots,\tau_n-1}$. 
(We elaborate on the physical meaning of $W_n(\cdot)$ later in Section \ref{se:Whittle index Section}). 

\begin{lemma}\label{th:w-subsidy optimal policy/value lemma}
Consider the $\omega$-subsidy problem \eqref{eq:w-subsidy problem} for project $n$. Then,
\begin{enumerate}
\item $\sigma_n(0,0)$ is optimal iff the subsidy $\omega\leq W_n(0)$.
\item For $\theta\in \{1,\cdots,\tau_n-1\},~\sigma_n(\theta,0)$ is optimal iff the subsidy $\omega$ satisfies  $W_n(\theta-1)\leq\omega\leq W_n(\theta)$.
\item $\sigma_n(\tau_n,0)$ is optimal iff $\omega=W_n(\tau-1)$. 
\item $\sigma_n(\tau_n,1)$ is optimal iff $\omega\geq W_n(\tau-1)$.
\end{enumerate}
In addition, for $\theta \in \{0,1,\ldots,\tau\}$, the policies $\{\sigma_n(\theta,\rho): \rho \in [0,1] \}$ are optimal when, 
\begin{enumerate}[i)]
\item $0\leq \theta \leq \tau-1$ and $\omega=W_n(\theta)$,
\item $\theta=\tau$ and $\omega=W_n(\tau-1)$.
\end{enumerate}
Furthermore, for any $\theta \in \{0,\cdots,\tau\}$, under the  $\sigma(\theta,0)$ policy, the average reward earned is, 
\begin{align}\label{rs1}
\frac{p_n\theta \omega-\eta E_n-(1-p_n)^{\tau_n-\theta}}{1+\theta p_n}.
\end{align}
Meanwhile, under the $\sigma_n(\tau_n,1)$ policy, the reward is $\omega-1$. 
\end{lemma}
\begin{proof}
For the $\omega$-subsidy problem of project $n$, 
let us first analyze the $\sigma_n(\theta,0)$ policy. The subscript $n$ is suppressed in the following. 
For each $\theta \in \{0,1,\cdots,\tau\}$, $\sigma(\theta,0)$ is a deterministic stationary policy. 
That is, for each $\sigma(\theta,0)$,  there exists a function $g(\cdot)$ defined on the state space $ \{0,1,\cdots,\tau\}$ of the project, such that $U_n(t)=g(Y_n(t))$.  
Further, there exist a real number $R$ and a real function $f$ on the state space with $f(0)=0$ such that,
\begin{align*}
R+f(i)&=-\bm{1}\left\{i=\tau \right\}-  g(i)E \eta+\omega \left(1-g(i)\right)\\
&+
p g(i)f(0)
+(1-p)g(i)f\Big((i+1)\wedge \tau \Big) \\
&+
\Big(1-g(i)\Big) f\Big((i+1)\wedge \tau \Big), \forall i=0,1,\cdots,\tau.
\end{align*}
The value of $R$ and $f(i),i=1,\cdots,\tau$ can be obtained by solving the $\tau\!+\!1$ equations above, and it can be shown that the $R$ is the average expected system reward under this $\sigma(\theta,0)$ policy (see \cite{Puterman1994}). 
Then, by standard results in infinite-horizon dynamic programming, see \cite{Puterman1994}, policy $\sigma(\theta,0)$ is optimal if and only if the following optimality equation is satisfied, 
\begin{align}\label{eq:w-subsidy optimality equation}
\nonumber
R+f(i)&= \max_{u\in\{0,1\}}\Big\{
-\bm{1}\{i=\tau\}-u E\eta +\omega (1-u)\\ \nonumber
&+ p u f(0)+(1-p)u f\Big((i+1)\wedge \tau \Big) \\
&+(1\!-\!u)f\Big((i\!+\!1)\wedge \tau \Big)
\Big\},\forall i=0,\cdots,\tau.
\end{align}
Similar results hold for the policy $\sigma(\tau,1)$, under which the system is always passive. The conditions in 1)-4), and the average expected system reward under these policies are obtained.

To obtain the conditions i) and ii), note that $\sigma(\theta,0)=\sigma(\theta+1,1)$, and the policy $\sigma(\theta,\rho),\rho\in(0,1)$ can be regarded as a combination of $\sigma(\theta,0)$ and $\sigma(\theta,1)$. 
\end{proof}



\begin{theorem}\label{th:primal-dual theorem}
For the relaxed problem \eqref{eq:relaxed-OP object}-\eqref{eq:relaxed-OP constraint} and its dual $d(\omega)$, the following results hold:
\begin{enumerate}
\item The dual function $d(\omega)$ satisfies,
\begin{align*}
d(\omega)=\sum_{n=0}^{N-1} R_n(\omega)-\omega(1-\alpha)N.
\end{align*}
\item Strong duality holds, i.e., the optimal average reward for the relaxed problem, denoted $R_\text{rel}$, satisfies,  
\begin{align*}
R_\text{rel}=\min_{\omega\geq 0} d(\omega).
\end{align*}
\item  In addition, $d(\omega)$ is a convex and piecewise linear function of $\omega$. Thus, the value of $R_\text{rel}$ can be easily obtained.
\end{enumerate}

\end{theorem}
\begin{proof}
For 1), it follows from Lemma \ref{th:w-subsidy optimal policy/value lemma} that for the $\omega$-subsidy problem associated with each project $n$, there is at least one stationary optimal policy, and under this policy, the optimality equation holds true. 
Thus, under the optimal policy, the limit of the time average reward exists (which is closely related to the optimality equation, see \cite{Puterman1994}). 
That is, the $\limsup_{T\rightarrow+\infty}$ in \eqref{eq:w-subsidy problem} can be replaced by $\lim_{T\rightarrow+\infty}$. 
As a result, all the ``less than or equal to" in \eqref{eq:dual function upper bound} can be replaced by equality signs. This proves the first statement.

For 2), the strong duality is proved by showing complementary slackness. The details are omitted due to space constraints.

For 3), it follows from equation~\eqref{rs1} that each $R_n(\omega)$ is a piecewise linear function. To prove convexity of $R_n(\omega)$, note that the reward earned by any policy is a linear function of $\omega$, and the supremum of linear functions is convex.
Thus, by statement 1), $d(\omega)$ is also convex and piecewise linear. In addition, since each $R_n(\omega)$ can be easily derived from Lemma \ref{th:w-subsidy optimal policy/value lemma}, the expression of $d(\omega)$ easily follows. Thus, $R_\text{rel}$, which is the minimum value of this known, convex, and piecewise linear function $d(\omega)$, can be easily obtained.
\end{proof}

\section{The Large Client Population Asymptotic Optimality of the Index Policy}\label{se:Whittle index Section}

The \textit{Whittle index} (see \cite{Whittle1988}) $W_n(i)$ of project $n$ at state $i$ is defined as the value of the subsidy that makes the passive and active actions equally attractive for the $\omega$-subsidy problem associated with project $n$ in state $i$. 
The $n$-th project is said to be \textit{indexable} if the following is true: Let $B_n(\omega)$ be the set of states for which project $n$ would be passive under an optimal policy for the corresponding $\omega$-subsidy problem. Project $n$ is \textit{indexable} if, as $\omega$ increases from $-\infty$ to $+\infty$, the set $B_n(\omega)$ increases monotonically from $\emptyset$ to the whole state space of project $n$. The bandit problem is indexable if each of the constituent projects is indexable. 

\begin{lemma}\label{th:whittle indexability Lemma}
The following are true:
\begin{enumerate}
\item The Whittle index $W_n(i)$ of project $n$ at state $i$ is, 
\begin{align*}
W_n(i)=
p_n(i+1)(1-p_n)^{\tau_n-(i+1)}-\eta E_n, 
\end{align*}
when $i=0,1,\cdots,\tau_n-1$; while $W_n(\tau_n)=W_n(\tau_n\!-\!1)$.
\item The stringent-constraint scheduling problem  \eqref{eq:original-OP object}-\eqref{eq:original-OP constraint} is indexable.
\item For each project $n$, the transition rates of its states in the associated $\omega$-subsidy problem form a unichain (there is a state $j\in\{0,1,\cdots,\tau_n\}$ such that there is a path from any state $i\in \{0,1,\cdots,\tau_n \}$ to state $j$), regardless of the policy employed. 
\end{enumerate}
\end{lemma}

\begin{proof}
Statements $1$) and $2$) directly follow from Lemma \ref{th:w-subsidy optimal policy/value lemma} and the definition of Whittle index, indexability. 
To prove statement $3$), note that since $p_n<1$, there is a positive probability that there is no packet delivery for $\tau_n$ successive time slots, regardless of the policy employed. Thus, from any state $i\in\{0,1,\cdots,\tau_n \}$, there is a path to the state $\tau_n$.
\end{proof}

As a result, 
the Whittle indices induce a well-defined order on the state values of each project. This gives the following heuristic policy.

\textit{Whittle Index Policy}: At the beginning of each time slot $t$, client $n$ is scheduled if its Whittle index $W_n\left(Y_n\left(t\right)\right)$ is positive, and, moreover, is within the top $\alpha N$ index values of all clients in that slot. Ties are broken arbitrarily, with no more than $\alpha N$ clients simultaneously scheduled. 

Now, we show the asymptotic optimality property of the Whittle Index Policy.  Classify the $N$ projects into $K$ \textit{classes} such that the projects in the same class have the same values of $p_n$, $\tau_n$ and $E_n$, while projects not in the same class differ in at least one of these parameters. For each class $k\in\{1,\cdots,K\}$, denote by $\gamma_k$ the \textit{proportion} of total projects that it contains; that is, there are $\gamma_k N$ projects in class $k$.

\begin{assumption}\label{global attractor assumption}
Construct the \textit{fluid model} of the restless bandit problem \eqref{eq:original-OP object}-\eqref{eq:original-OP constraint} as in \cite{Weber1990} and  \cite{Verloop2014}, and denote the fluid process as $\mathbf{z}(t)$. We assume that, under the Whittle Index Policy, $\mathbf{z}(t)$ satisfies the global attractor property. That is, there exists $\mathbf{z}^\star$ such that from any initial point $\mathbf{z}(0)$, the process $\mathbf{z}(t)$ converges to the point $\mathbf{z}^\star$, under the Whittle Index Policy. 
\end{assumption}

This assumption is not restrictive because of the following:  
First note that the MDP-2 itself also satisfies the unichain property. 
Then, under the Whittle Index policy, MDP-2 also forms a unichain. As a result, this $N$ client bandit problem has a single recurrent class, and has a global attractor. Thus, it is not restrictive to assume that its fluid model also  satisfies the global attractor property. 

\begin{theorem}
When Assumption \ref{global attractor assumption} holds, as the number $N$ of clients increases to infinity, $R_\text{ind}/N \rightarrow R_\text{rel}/N $, where $R_\text{ind}$ and $R_\text{rel}$ is the system reward under the Whittle Index policy and the optimal relaxed policy, respectively. 
(Here, the fraction of active bandit $\alpha$ and the proportion of each bandit class $\gamma_k$ remain the same when $N$ increases. In addition, 
the client number $N$ is such that all  $\gamma_k N$ are  integers.) Thus, the Whittle Index policy is asymptotically optimal. 
\end{theorem}

\begin{proof}
By Assumption \ref{global attractor assumption} and Lemma \ref{th:whittle indexability Lemma}, $R_\text{ind}/N \rightarrow R_\text{rel}/N $ directly from the result in \cite{Verloop2014}. Note that $R_\text{rel}$ is an upper-bound for the stringent-constraint problem; thus, the asymptotic optimality holds.
\end{proof}

\section{Simulation Results}

We now present the results of simulations of Whittle Index policy with respect to its average cost per client. The numerical results of the relaxed-constraint problem \eqref{eq:relaxed-OP object}-\eqref{eq:relaxed-OP constraint}, which is derived by Theorem \ref{th:primal-dual theorem} and Lemma \ref{th:w-subsidy optimal policy/value lemma}, are also employed to provide a bound on the stringent-constraint problem.  

Fig. \ref{fg:asymptical optimal} illustrates the average cost per client under the relaxed optimal policy and the Whittle Index policy for different total numbers of clients. 
It can be seen that when the total number of clients increases, the gap between the relaxed optimal cost and the cost under the Whittle Index policy shrinks to zero.
Since the optimal cost of the relaxed-constraint problem serves as a lower bound on the cost in the stringent-constraint problem,  
this means the Whittle Index policy approaches the optimal cost as the total number of clients increases, i.e., the Whittle Index policy is asymptotically optimal.  

\begin{figure}[!t]
	\centering
	\includegraphics[width=0.5\textwidth]{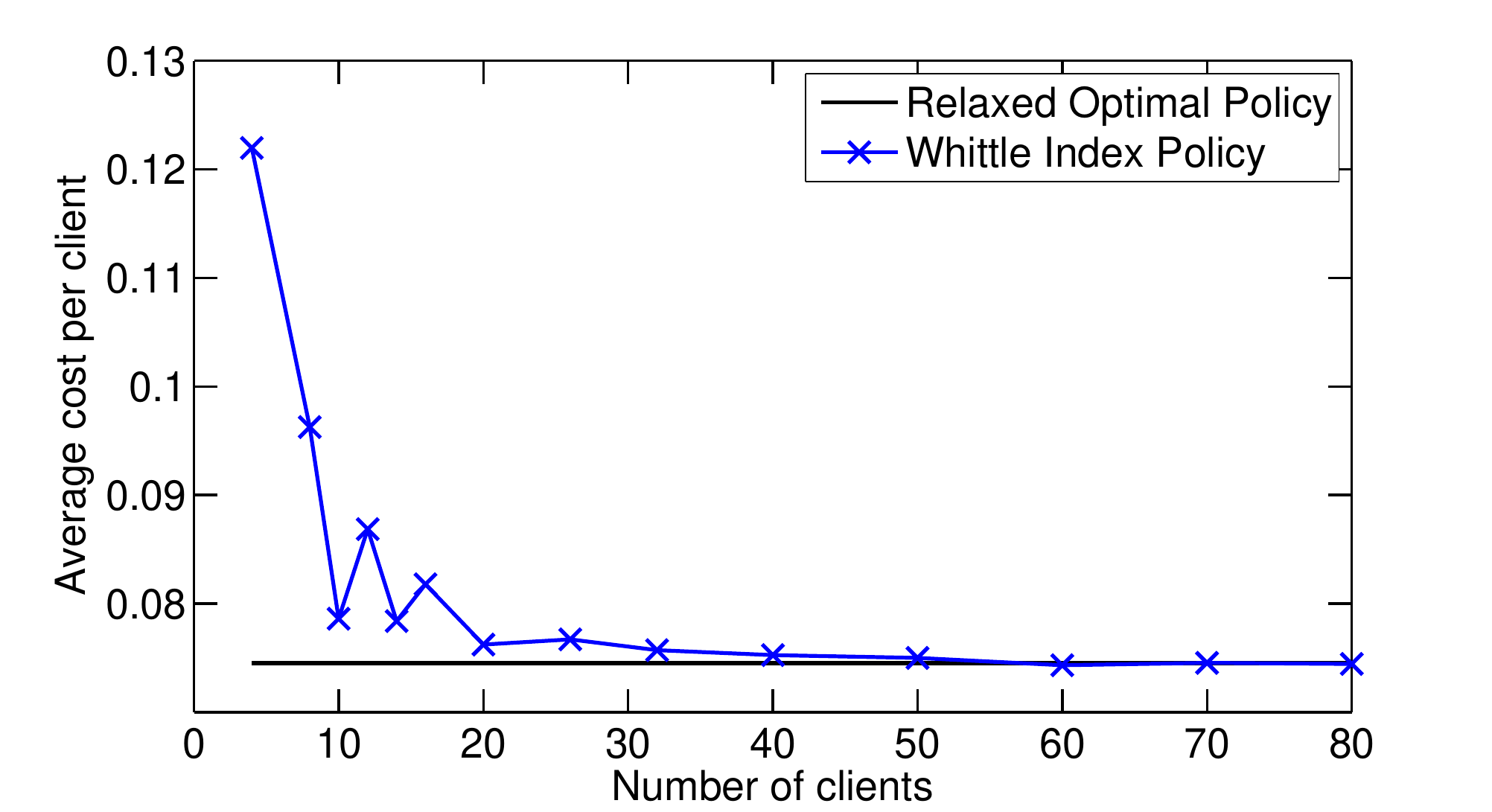}
	\caption{The time average cost per client vs. the total number of clients for the optimal policy under the relaxed constraint and the Whittle Index policy are shown. (The parameters are $\alpha=0.3$, $\eta=0.1$, with $K=2$ classes of projects, and  $\gamma_1=0.5$, $\gamma_2=0.5$ proportion of projects in each class. For each client $n$ in the first class, $p_n=0.6$, $\tau_n=10$, $E_n=2$; while for each client $n$ in the second class, $p_n=0.8$, $\tau_n=5$, $E_n=3$.)}
	\label{fg:asymptical optimal}
\end{figure}

Fig. \ref{fg:eta effect} illustrates the average inter-delivery penalty per client versus the average energy consumption per client under the Whittle Index policy for different values of the energy-efficiency parameter $\eta$. As $\eta$ increases, the average energy consumption decreases, while the average inter-delivery penalty increases. Thus, there is a tradeoff between energy-efficiency and inter-delivery regularity. By changing $\eta$, we can balance these two important considerations.

\begin{figure}[!t]
	\centering
	\includegraphics[width=0.5\textwidth]{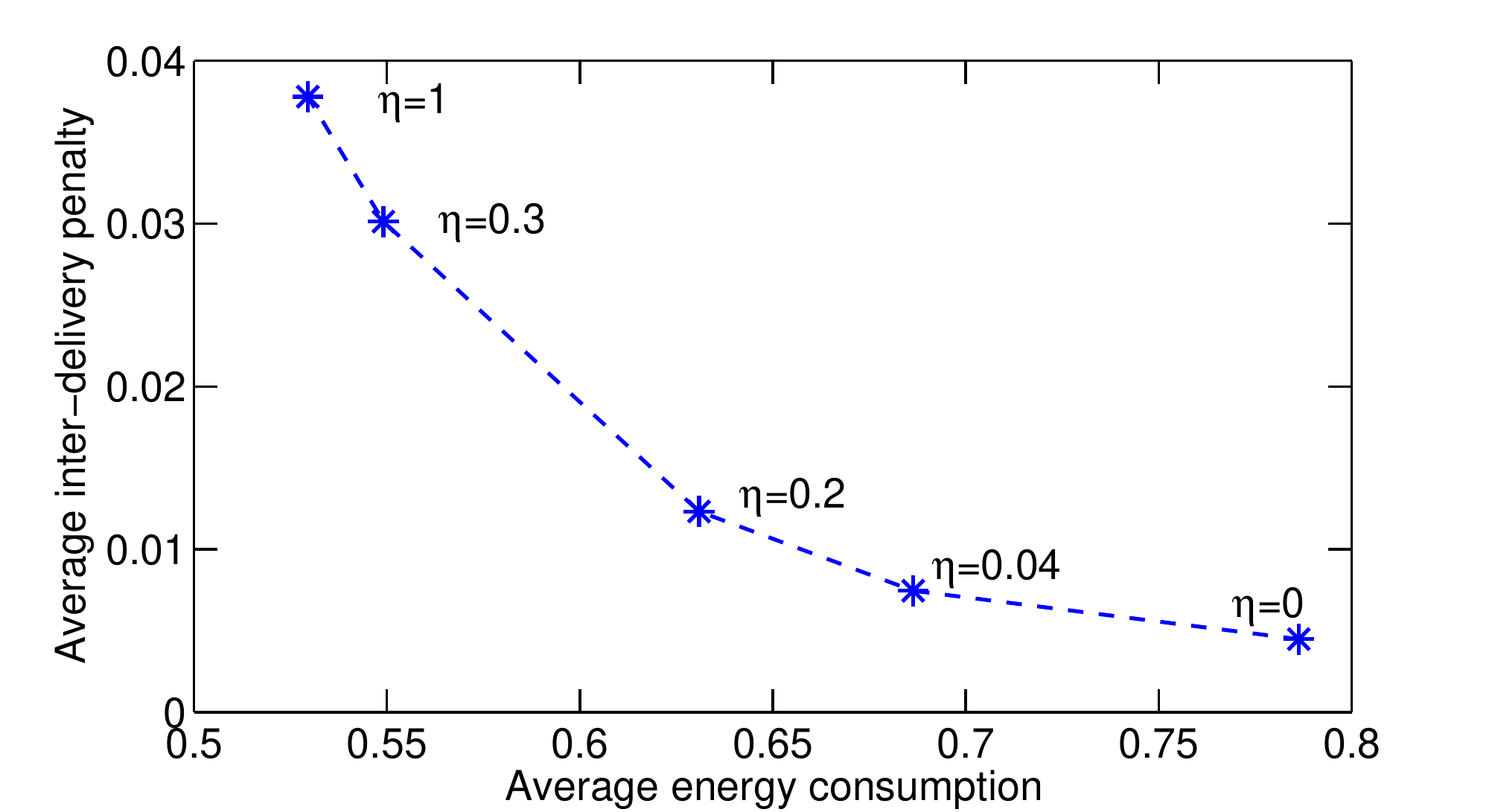}
	\caption{The time average inter-delivery penalty per client vs. the time average energy consumption per client under the Whittle Index Policy for different values of energy-efficiency parameter $\eta$ are shown. (The parameters are $N=100$, $\alpha=0.3$, with $K=2$ classes of projects, and $\gamma_1=0.5,~\gamma_2=0.5$ proportion of projects in each class. For each client $n$ in the first class, $p_n=0.6$, $\tau_n=10$, $E_n=2$; while for each client $n$ in the second class, $p_n=0.8$, $\tau_n=5$, $E_n=3$.)}
	\label{fg:eta effect}
\end{figure}

\section*{Acknowledgment}
This work is sponsored
in part by the National Basic Research Program of China
(2012CB316001), the Nature Science Foundation of China
(61201191, 60925002, 61021001), NSF under Contract Nos.
CNS-1302182 and CCF-0939370, AFOSR under Contract No.
FA-9550-13-1-0008, and Hitachi Ltd.

\bibliographystyle{IEEEtran}
\bibliography{SIP}
\end{document}